\documentclass[aps,pre,twocolumn,superscriptaddress,longbibliography]{revtex4-2}

\usepackage{amsmath, amsthm}
\usepackage{amssymb}
\usepackage{graphicx}
\usepackage{braket}
\usepackage{lipsum}
\usepackage{xcolor}
\usepackage{epsfig}
\usepackage{enumitem}
\usepackage{mathtools}
\usepackage[implicit=false]{hyperref}
\usepackage{lipsum}

\usepackage{cleveref}

\newtheorem{theorem}{Theorem}[section]

\newtheorem{cor}[theorem]{Corollary}

\theoremstyle{definition}
\theoremstyle{definition}

\newcommand{\R}{\mathbb{R}}

\newcommand{\fref}[1]{Fig.~\ref{#1}}
\newcommand{\Fref}[1]{Figure~\ref{#1}}

\newcommand{\eref}[1]{(\ref{#1})}
\newcommand{\Sref}[1]{Section~\ref{#1}}
\newcommand{\sref}[1]{Sec.~\ref{#1}}

\newcommand{\lp}{\left}
\newcommand{\rp}{\right}

\begin{document}

\preprint{APS/123-QED}

\title{Bifurcations and multistability in empirical mutualistic networks }

\author{Andrus Giraldo}%
\email{agiraldo@kias.re.kr}
\author{Deok-Sun Lee}
\email{deoksunlee@kias.re.kr}
\affiliation{School of Computational Sciences, Korea Institute for Advanced Study, Seoul 02455, Korea}

\date{\today}

\begin{abstract}
Individual species may experience diverse outcomes, from prosperity to extinction, in an ecological community subject to external and internal variations. Despite the wealth of theoretical results derived from random matrix ensembles, a theoretical framework still remains to be developed to understand species-level dynamical heterogeneity within a given community, hampering real-world ecosystems' theoretical assessment and management. Here, we consider empirical plant-pollinator mutualistic networks,  additionally including all-to-all intragroup competition, where species abundance evolves under a Lotka-Volterra-type equation. Setting the strengths of competition and mutualism to be uniform, we investigate how individual species persist or go extinct under varying the interaction strengths. By employing bifurcation theory in tandem with numerical continuation, we elucidate transcritical bifurcations underlying species extinction and demonstrate that the Hopf bifurcation of unfeasible equilibria and degenerate transcritical bifurcations give rise to multistability, i.e., the coexistence of multiple attracting feasible equilibria. These bifurcations allow us to partition the parameter space into different regimes, each with distinct sets of extinct species,  offering insights into how interspecific interactions generate one or multiple extinction scenarios within an ecological network. 
\end{abstract}

\maketitle

\newpage

\section{Introduction} \label{sec:Int}

The stability of biodiversity and species abundance distribution in each ecological community - whether maintaining current states or undergoing abrupt changes, including the extinction of some or almost all species - is a crucial question in ecological conservation efforts and also has been of central interest to theoretical approaches in ecology~\cite{May1972,ROBERTS:1974aa, Takeuchi1996}. Insights have been gained from studying the properties of equilibrium states in an ensemble of fully connected communities with random interaction strength, utilizing tools e.g., from the spin glass physics and random matrix theory~\cite{HRieger_1989,Allesina2012,*Allesina2015,Bunin2017,GarciaLorenzana2022,Gibbs2018,Mambuca2022}. Yet, real-world communities are often sparsely connected and heterogeneously structured~\cite{bascompte2003nested, thebault2010stability, suweis2013emergence,PhysRevLett.108.108701,stone2016google}, demanding further investigations to adapt these theoretical frameworks to empirical features~\cite{park2024incorporating,poley2024interaction}. Moreover, while ensemble-averaged properties provide broad insights, they often fall short of predicting specific future scenarios for individual communities, which are crucial for effective assessment and management. This gap underscores the importance of developing new frameworks that are better suited to understand and predict the dynamics of specific communities, based on their actual structures.

Under diverse environmental and internal perturbations, individual species can experience variations in their abundances or face extinction, and various theoretical approaches have been proposed to explore these dynamics~\cite{MacMahon1982, PRL2006Assaf, volkov2003neutral,suweis2023generalized}. Simple dynamical models, such as Lotka-Volterra(LV)-type equations, have been widely used because they effectively capture the evolution of species abundance under nonlinear interspecific interactions~\cite{DeOca1996,Baigent2017,Goh1979, Smale1976,Takeuchi1996,Zeeman1995,Zeeman2002,Zeeman2002a}. Recent empirical data-sets from real-world ecological networks~\cite{webOfLife,montoya2002small,dunne2002food,pascual2006ecological} have spurred investigations into such model dynamics on these networks, offering insights into the interplay between structure and dynamics~\cite{bastolla2009architecture, rohr2014structural, Maeng_2019,Lee2022}. 

Identifying equilibria, periodic solutions, and their stability from the nonlinear dynamics in specific real-world networks can enable predictions of possible specific scenarios for their biodiversity and species composition~\cite{marcus2024extinctions}. This may be a key next-generation goal in the field of nonlinear dynamics, especially in the era of big data. However, it can be challenging because real-world systems are typically large and their dynamics are high dimensional,  making mathematical techniques of lower dimensional systems harder to apply\cite{Aguirre2021, https://doi.org/10.1111/ele.14413,strogatz2018nonlinear}. To proceed, various approximate methods have been attempted, such as approximating the interspecific interaction matrix by a low-rank matrix to gain insights into the attracting equilibrium, assumed to be unique,  in specific network structures~\cite{stone2016google,Lee2022}. However, even small errors in these predictions can lead to disastrous results when applied to real-world ecosystems, making a methodology that is as exact as possible highly desirable. 

Given this background, we apply established methods from bifurcation theory \cite{KuznetsovBook2004} and numerical continuation~\cite{Doedel2007} to
investigate the equilibria that are both attracting and feasible, with the latter meaning no negative abundance, and how they vary with system parameters in empirical plant-pollinator networks governed by a LV-type equation. While keeping the real-world sparse and heterogeneous structure of mutualistic partnership between two groups - plant and pollinator - of species~\cite{bascompte2003nested, thebault2010stability, suweis2013emergence,PhysRevLett.108.108701,stone2016google}, we assume uniform mutualism strength~\cite{ROBERTS:1974aa,stone2016google,Lee2022} and all-to-all uniform intragroup competition that may arise from limited resources~\cite{suweis2013emergence,Pascual-Garcia:2017wj}.

We find that among exponentially many equilibria, only one is both feasible and attracting unless interactions are sufficiently strong; As interaction strengths increase, different equilibria successively become the only globally attracting equilibrium through \textit{transcritical bifurcations} each featuring no or nonempty sets of extinct species. When interactions are sufficiently strong, different initial conditions can lead to different equilibria, so a species may survive in one but be extinct in another equilibrium. Our study reveals that such multistability arises from two distinct types of bifurcations: \textit{Hopf bifurcation} and \textit{degenerate transcritical bifurcation}. While feasible equilibria cannot exhibit Hopf bifurcations in our studied networks due to them having symmetric interaction matrices ~\cite{MacArthur1970,PhysRevLett.130.257401}, our study demonstrates that the Hopf bifurcation of an unfeasible equilibrium, followed by a transcritical bifurcation, can lead to the creation of a new feasible attracting equilibrium. Also, degenerate transcritical bifurcations are shown to create attracting equilibria with different components in a particular number of surviving species from the original equilibrium, resulting in similar biodiversity scenarios. In contrast, Hopf bifurcations can lead to vastly different biodiversity scenarios. Our study thus elucidates the geometrical mechanisms underlying species extinction and multistability in empirical mutualistic networks. By identifying these bifurcations, we determine the exact parameter regimes that display distinct sets of surviving species. This can provide a platform for developing strategies  to preserve and control biodiversity in real ecological communities. 

The paper is organized as follows. In \sref{set:Model}, we introduce the model 
and the persistence diagrams illustrating the survival and extinction of individual species and multistability in the empirical networks under study. In Sec.~\ref{sec:Method}, we provide the mathematical properties of the model and  the  numerical continuation methods to identify potential bifurcations. In Sec.~\ref{sec:Transcritical}, we explore the transcritical bifurcations that underlie species extinction. The two geometrical mechanisms for multistability are investigated in detail in Sec.~\ref{sec:Hopf} and ~\ref{sec:degTrans}, respectively. In Sec.~\ref{sec:BifurcationDiagram}, we present a phase diagram illustrating  how the bifurcations organize themselves as the strengths of competition and mutualism are varied. Finally, we conclude in Sec.~\ref{sec:Conclusion} with a summary of our findings and an outlook on open questions and future research.

\section{Model and overview}
\label{set:Model}

We consider the abundances $x_1,x_2, ... , x_s$ of $s$ species, including $n_p$ plant-group species and $n_a$ animal(pollinator)-group species.  For the remainder of the paper, we use indices $1$ to $n_p$ for plant species and $n_p+1$ to $s=n_p+n_a$ for animals. Species interact via  all-to-all intragroup competition, and selective intergroup mutualism represented by the mutualism adjacency matrix $A$ of dimensions $n_p \times n_a$ with $A_{ij}=A_{ji}=1$ if plant species $i$ and animal species $n_p+j$ are in mutualistic relation, or zero otherwise. These interspecific interactions affect the species abundance evolving with time as 
\begin{equation} 
\label{eq:LVNet}
\dfrac{ d\mathbf{x}}{dt} = f(\mathbf{x})=\mathbb{X} \lp( \boldsymbol{\alpha} + \mathbb{B}\mathbf{x} \rp),
\end{equation}
where 
$\mathbf{x}:=(x_1,x_2, ... , x_s)$ is the vector of species abundance, $\boldsymbol{\alpha}:=\alpha(1,1, ... , 1)$ is a vector of self-growth rates with $\alpha$ a positive constant, and $\mathbb{X}:= {\rm diag}(x_1,x_2, ... , x_s)$ is a diagonal matrix with the species abundance.  The interaction matrix $\mathbb{B}$ represents the self-regulation, the intragroup competition, and the intergroup mutualism by 
\begin{align}
\mathbb{B} &:=-(1-c)\mathbb{I} - c\mathbb{J}+m \mathbb{A}\nonumber\\
&=
\begin{pmatrix}
    -1 & -c& \cdots& -c&  &  & &   \\
    -c & -1 & \cdots& -c&  & &  & \\
    \vdots &  & & \vdots& &  m\,A&  & \\
    -c & \cdots & -1 & -c & & & & \\
    -c & \cdots & -c & -1 & & & & \\
    & & & & -1 & -c& \cdots& -c\\
    & & & & -c & -1& \cdots& -c\\
    & m\,A^T& & & \vdots &  & & \vdots\\
    & & & & -c & \cdots & -1 & -c \\
    & & & & -c & \cdots & -c & -1
\end{pmatrix},
\label{eq:Bmatrix}
\end{align}
where $\mathbb{I}$ is the identity matrix of dimension $s\times s$, and  
$$\mathbb{J} := \begin{pmatrix}
\mathbf{J}_{n_p} & 0 \\
0 & \mathbf{J}_{n_a} 
\end{pmatrix}
\quad {\rm and} \quad \mathbb{A} := \begin{pmatrix}
0 & A \\
A^T & 0 
\end{pmatrix}$$
with $\mathbf{J}_{n_p}$ and $\mathbf{J}_{n_a}$ the all-ones matrices of dimensions $n_p\times n_p$ and $n_a\times n_s$, respectively.
Notice that parameters $c$ and $m$ control the strength of competition and mutualism, respectively. We set $\alpha=1$ as $\alpha$ only rescales time and abundance; that is, if $\mathbf{x}(t)$ is a solution for a particular $\alpha$ then $\alpha\mathbf{x}(\alpha t)$ is a solution to Eq.~\eref{eq:LVNet} with $\alpha=1$. Without interspecific interactions, i.e., $c=m=0$, each species would be independent. 

\begin{figure*}
\centering
\includegraphics{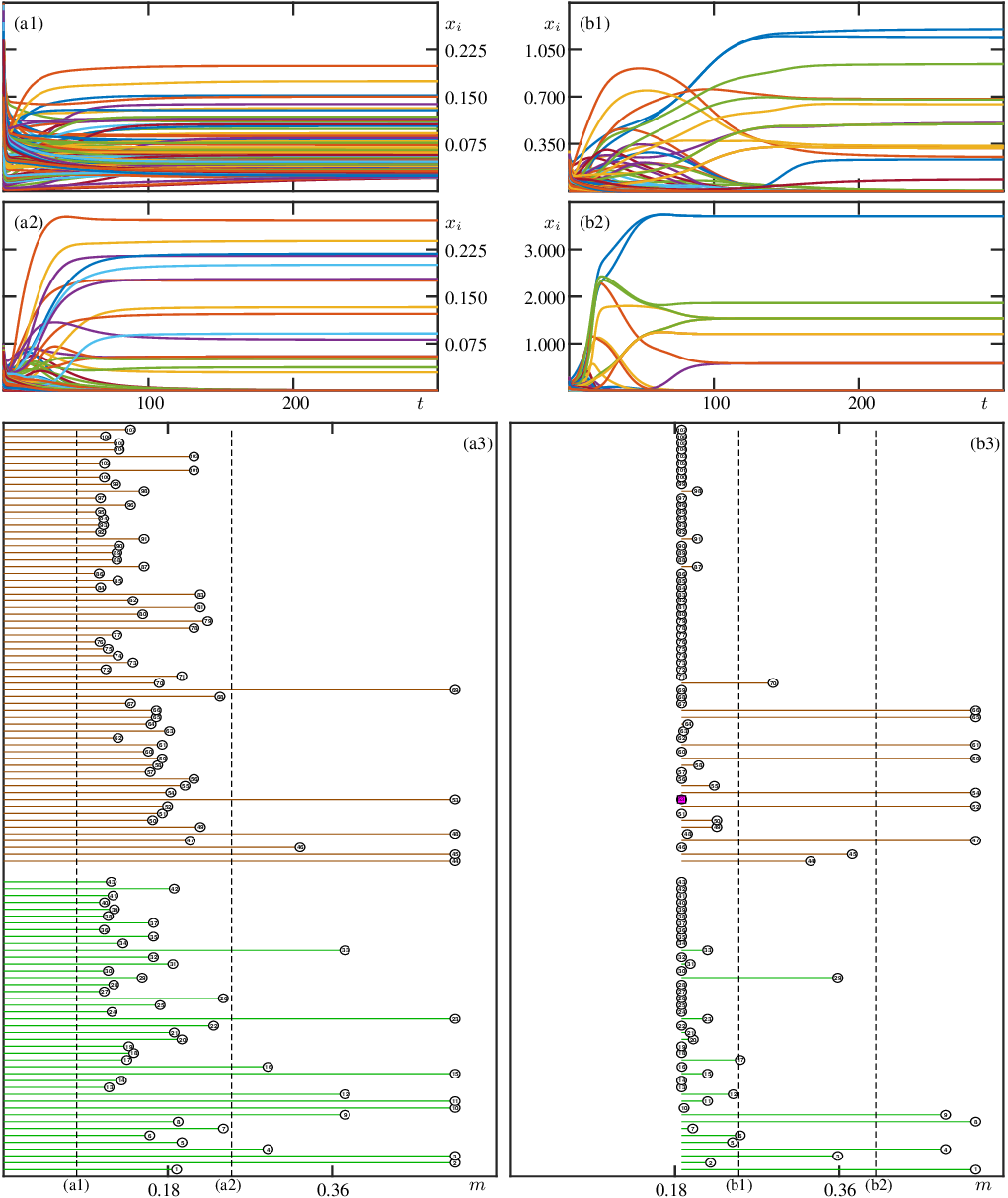}
\caption{Abundance and persistence of individual species.
(a1) Time evolution of the abundances $x_i$'s of 107 species in network $A_1$  with $c=0.3$ and $m=0.08$ for a randomly-selected initial condition. 
(a2) The same as (a1), but with $m=0.25$.
(a3) Persistence diagram, 
representing by horizontal lines, green for plants and brown for animals, the ranges of $m$ for which species survive when $c$ is fixed at $0.3$. This is obtained by numerical continuation, starting from the full-coexistence equilibrium at $m=0$.   
(b1) The same as (a2), but for a different initial condition.
(b2) The same as (b1), but with $m=0.485$. 
(b3) Persistence diagram starting from another equilibrium at $m=m^{(2)}\approx 0.187267$. 
The species 53 represented by a purple square has a negative abundance before a transcritical bifurcation makes it zero at $m^{(2)}$. See Sec.~\ref{sec:Hopf}. Vertical dashed lines in (a3) and (b3) indicate the values of $m$ used in (a1),(a2), (b1), and (b2). } 
\label{fig:BifNetwork1} 
\end{figure*}

For the mutualism matrix $A$, we use two datasets  from the Web of Life database~\cite{webOfLife}, one consisting of $43$ plants and $64$ pollinators, and the other $29$ plants and $81$ pollinators, which we refer to as networks $A_1$ and $A_2$, respectively. The structure of these real-world mutualistic communities is not uniform or random but structured such that they are often nested~\cite{bascompte2003nested,Lee:2012aa,PhysRevX.9.031024} and individual species have widely different numbers of mutualistic partners, resulting in different numbers of non-zero elements across rows and columns in the matrix $A$ in Eq.~\eqref{eq:Bmatrix}~\cite{suweis2013emergence,PhysRevLett.108.108701}. The aim of the present study is to investigate the equilibrium states of such structured competitive-mutualistic networks.

Setting the initial abundance $x_i(0)$ of each species $i$ to be a random number between $0$ and $1$, we numerically integrate system~\eqref{eq:LVNet} with a fixed competition strength $c=0.3$ and different values of the mutualism strength $m$ to obtain $x_i(t)$'s for all $i$ and time $t$ as shown in \fref{fig:BifNetwork1}(a1),(a2), (b1), and (b2). When the mutualistic interaction is sufficiently weak, e.g., $m=0.08$ as in \fref{fig:BifNetwork1}(a1),  all the species coexist, which we call the full-coexistence state. On the other hand, for $m=0.25$, only $17$ species survive as shown in \fref{fig:BifNetwork1}(a2); hence, strong mutualistic interactions may break the full coexistence, resulting in some species becoming extinct while others flourish. 

The persistence of species under a gradual increase of $m$, modeling the environmental or internal perturbations, exhibits a strong inequality; the persistence line drawn in the range $[0, m^*_{F,i})$, for which each species $i$ survives, is of quite different length. It is also remarkable that no line extends infinitely, as the surviving species' abundance grows without bounds past the end point. These persistence lines are computed by starting a numerical continuation scheme from the attracting equilibrium of system~\eref{eq:LVNet} at $m=0$, and then studying its variation under increments of $m$; more details are in the next section.

At large values of $m$, one can find a totally different equilibrium state,  by using different initial conditions, from that reached from the full-coexistence state considered in Fig.~\ref{fig:BifNetwork1}(a3). This means that there can be multiple attracting equilibria, for the same system parameters, each representing a different scenario of species survival and extinction. For instance, at $c=0.3$ and $m=0.3$, the time-evolution of individual species abundance is different between \fref{fig:BifNetwork1}(a2) and \fref{fig:BifNetwork1}(b1). The final state in \fref{fig:BifNetwork1}(b2) at $m=0.485$ showcases how the new attracting equilibrium in \fref{fig:BifNetwork1}(b1) exhibits different extinctions as $m$ increases compared to \fref{fig:BifNetwork1}(a3). In \fref{fig:BifNetwork1}(b3), we show the persistence diagram of the new equilibrium, which first  appears at $m^{(2)} \approx 0.187267$ with some species already extinct at that point.

The exact persistence diagrams in \fref{fig:BifNetwork1}(a3) and (b3) extend the approximate results in \cite{Lee2022} for small values of $m$. In the subsequent sections, we describe the numerical method used to generate these diagrams and investigate the variation of the equilibrium states with $c$ and $m$ from the perspective of bifurcation theory. The identified nature of the bifurcations underlying species extinction and the appearance of multiple attracting equilibria provides a geometrical picture of the equilibria and their stability.

\section{Numerical continuation method to trace the attracting and feasible equilibria}
\label{sec:Method}

The computational cost of integrating numerically system~\eref{eq:LVNet} with continuously varying the system parameters, to obtain the persistence diagrams shown in Fig.~\ref{fig:BifNetwork1}(a3) and (b3), is quite high; especially, if one would want a good resolution in both $m$ and $c$. Instead, one can consider the equilibria  ${\bf x}^*$ of system~\eqref{eq:LVNet} satisfying 
\begin{equation} 
\label{eq:LVNetZeros}
\mathbb{X}^* \lp( \mathbf{1}_s + \mathbb{B}\mathbf{x}^* \rp) =0
\end{equation}
and their stability. Once identifying an attracting equilibrium for particular values of $c$ and $m$, we employ pseudo-arch length continuation techniques \cite{Doedel2007} to study the variation in their components(abundances) and stability with $c$ and $m$. In this way, we can pinpoint the bifurcations that the equilibria might exhibit, allowing us to identify attracting equilibria for general values of the system parameters without a lot of additional cost.

System~\eref{eq:LVNet} can generically have exponentially many  equilibria given $2^s$ different scenarios of extinct species. With $\mathcal{J} \subseteq \{1, 2, ... , s\}$ a subset of indices, the equilibrium solution $\mathbf{x_{\mathcal{J}}^*} \in \R^s$ with the species in $\mathcal{J}$ extinct is given by
\begin{equation}
x_{\mathcal{J},j}^* = \left\{
\begin{array}{ll}
0 & \quad {\rm for} \quad j\in \mathcal{J},\\
-[\mathbb{B^+}^{-1}\mathbf{1}_{s-|\mathcal{J}|}]_j & \quad {\rm for} \quad j \notin \mathcal{J},
\end{array}
\right.
\label{eq:xequil}
\end{equation}
where $\mathbb{B^+}$ is the surviving species' interaction matrix obtained by eliminating in $\mathbb{B}$ the rows and columns that are in $\mathcal{J}$.  This equilibrium exists if and only if $\mathbb{B^+}$ is invertible; otherwise, an equilibrium 
with only $\mathcal{J}$ species extinct does not exist or there are infinitely many such equilibria. In the present study, it should be noted that not all but only feasible and attracting equilibria are physically meaningful. An equilibrium  with $x_i\geq 0$ for all $i$ is considered as feasible since species cannot have a negative abundance~\cite{ROBERTS:1974aa}, and as attracting, locally or globally, if certain or all initial conditions converge to it. A major challenge is thus to identify feasible and attracting equilibria and understand how they vary with the system parameters. 

At $m=0$ for $0 \leq c <1$,  the full-coexistence equilibrium $\mathbf{x_\emptyset^*}$ with no extinct species is globally attracting, from Goh's Theorem \cite{Goh1979,Baigent2017}, given that $\mathbb{B}$ is negative definite [Appendix~\ref{app:Goh}]. As $\mathbb{B}$ is symmetric  and $m\mathbb{A}$ can be seen as a continuous symmetric perturbation made to $\mathbb{B}$ from $m=0$, the eigenvalues of $\mathbb{B}$ vary continuously with $m$ and thus $\mathbb{B}$ remains negative definitive and $\mathbf{x_\emptyset^*}$ a global attractor for sufficiently small $m$. 

As $m$ becomes larger, two cases might arise: $\mathbb{B}$ stops being negative definitive, or $\mathbf{x_\emptyset^*}$ finds one or more components zero. In either case, Goh's theorem cannot be applied anymore, leading us to investigate the local attractivity of the full-coexistence or other equilibria. If all the eigenvalues of the Jacobian $Df|_{\mathbf{x^*_{\mathcal{J}}}}$ of system~\eref{eq:LVNet} at an equilibrium $\mathbf{x^*_{\mathcal{J}}}$ have negative real parts, then $\mathbf{x_{\mathcal{J}}}^*$ is stable (locally attracting). By studying the variation of these eigenvalues as $m$ increases, one can find when the full-coexistence equilibrium may become unstable and another equilibrium becomes stable. 

We apply numerical continuation techniques~\cite{Doedel1981, Doedel2010} to trace changes in the local stability of the equilibrium solutions satisfying Eq.~\eqref{eq:LVNetZeros} while parameters are varied. As $m$ increases for fixed $c$, an attracting equilibrium $\mathbf{x^*_{\mathcal{J}}}$ may become unstable at a certain value $m_*$, which can be detected by monitoring if the determinant of the Jacobian $Df |_{\mathbf{x^*_{\mathcal{J}}}}$ represented by~\cite{Lee2022} 
\begin{equation} 
\label{eq:stabFullDet}
\det \lp( Df |_{\mathbf{x^*_{\mathcal{J}}}} \rp) =   
\prod_{i\in \mathcal{J}} \lp( 1 + \sum_{\ell=1}^s B_{i\ell} x_{\mathcal{J},\ell}^*\rp)
\prod_{i \notin \mathcal{J}} \lp( x^*_{\mathcal{J},i} \rp) 
\det \lp(\mathbb{B}^+ \rp)
\end{equation}
becomes zero. For system~\eref{eq:LVNet}, we identify the following cases:
\begin{enumerate}
 \item The factor $\lp( \prod_{i\notin \mathcal{J}} x^*_{\mathcal{J},i} \rp)$ becomes zero, as one or more components not belonging to  $\mathcal{J}$ become zero, implying the extinction of the corresponding species, at $m_*$. If a surviving species, say $j$, has zero abundance at $m_*$,  this indicates a generic {\bf transcritical bifurcation} \cite{KuznetsovBook2004} as the equilibrium $\mathbf{x^*_{\mathcal{J}}}$, which is feasible and attracting when $m$ is close to but smaller than $m_*$, coincides with another equilibrium $\mathbf{x}^*_{\mathcal{J}^\prime}$ with $\mathcal{J}^\prime = \mathcal{J}\cup\{j\}$ at $m_*$. For $m > m_*$, $\mathbf{x^*_{\mathcal{J}}}$ is unstable (also unfeasible since $x^*_{\mathcal{J},j}<0$) and the new one $\mathbf{x}^*_{\mathcal{J}^\prime}$ is attracting and feasible; the species in $\mathcal{J}^\prime - \mathcal{J}$ are newly extinct 
 \footnote{If more than one components, say the components in a set  $\mathcal{K}$, become zero at $m_*$, then the equilibrium $\mathbf{x^*_{\mathcal{J}}}$  coincides with all the equilibria $\mathbf{x_{\mathcal{J}^\prime}^*}$ with all $\mathcal{J}^\prime$ in $\mathcal{P}(\mathcal{K})$, where $\mathcal{P}(\mathcal{K})$ is the power set of $\mathcal{K}$. This occurs if Eq.~\eref{eq:LVNet} is equivariant \cite{Golubitsky1985} under a permutation of species, i.e., when a collection of plants or animals have the same set of mutualistic partners.}.
  
 \item The factor $\det \lp(\mathbb{B}^+ \rp)$ is zero. Then $\mathbb{B}^+$ is not invertible, and the solution $\mathbf{x^*_{\mathcal{J}}}$ cannot be given by Eq.~\eqref{eq:xequil}; there are either infinitely many or no solutions at all. The former corresponds to  \textbf{degenerate transcritical bifurcations}, a mechanism generating multistability detailed in \Sref{sec:degTrans}.  The latter case implies that $\mathbf{x_{\mathcal{J}}^*}$ exhibits a \textbf{transcritical bifurcation at infinity} \cite{Giraldo2017, Giraldo2020I, Matsue2017}; Past this transition at $m_*$, one finds the new equilibrium with abundances growing unbounded.
\end{enumerate}
Increasing $m$ at fixed $c$, we compute which of the two cases occurs first. If the first one occurs, then we consider the new attracting equilibrium $\mathbf{x_{\mathcal{J}^\prime}^*}$ and check with it which of the two cases occurs while increasing $m$ further. These procedures are repeated until the second case occurs. Suppose that the second case occurs, and it is a degenerate transcritical bifurcation. Then, we consider at least two different equilibria, indicating multistability, and recursively study all the attracting equilibria while increasing $m$ further.  If the second case occurs and there is a transcritical bifurcation at infinity, then we stop. 

Since the considered networks are large, keeping track of all the equilibria in system~\eref{eq:LVNet} and finding which one can create multistatibility is challenging. To circumvent this problem, we numerically integrate system~\eref{eq:LVNet} for particular values of $m$ and $c$ under different initial conditions and investigate whether the solutions converge to more than one equilibria. If it happens, then we apply our continuation routine to each of those distinct equilibria in reverse, e.g., by decreasing $m$ at fixed $c$, to identify a series of transcritical bifurcations leading to the attracting and feasible equilibria with smaller sets of extinct species. Eventually 
an attracting but unfeasible equilibrium is identified at $m^\prime>0$, which is found to have been stabilized by \textbf{Hopf bifurcation} at $m$ smaller than $m^\prime$ as will be detailed in \Sref{sec:Hopf}. For each of these original equilibria, we produce persistence diagram as shown in \fref{fig:BifNetwork1}(a3) and (b3) to demonstrate multistability. 

These bifurcations disclose the geometrical mechanisms underlying species extinction and multistability; Hopf bifurcations and degenerate transcritical bifurcations lead to new attracting equilibria while transcritical bifurcations transfer stability from one equilirium to another, maintaining the number of attracting equilibria. The involved computations are performed using the widely-used software package \textsc{Auto07p} \cite{Doedel1981, Doedel2010} that finds and traces equilibria, periodic orbits, and suitable two-point boundary values problems as functions of system parameters. Moreover, it allows for the identification of different types of bifurcations, the monitoring of the stability of equilibria, and the determination of when two or more solution branches meet at a particular value and switch between them during continuation. The latter is relevant when different equilibria coincide at the moment of a transcritical bifurcation. A good introduction to such techniques and the use of \textsc{Auto07p} can be found in \cite{Doedel2007}.

\section{Results}
\label{sec:Results}

In this section,  we present and analyze the results obtained by applying the  method described in Sec.~\ref{sec:Method}. Numerical continuation yields the variation of an equilibrium, and the possible change of its stability as a parameter varies, which informs us of the nature of bifurcations and allows us to discover the geometrical mechanisms of species extinction and multistability. The results of this section lead to the phase diagram in Sec.~\ref{sec:BifurcationDiagram}.

\subsection{Transcritical bifurcations underlying species extinction}
\label{sec:Transcritical}

\begin{figure*}
\centering
\includegraphics{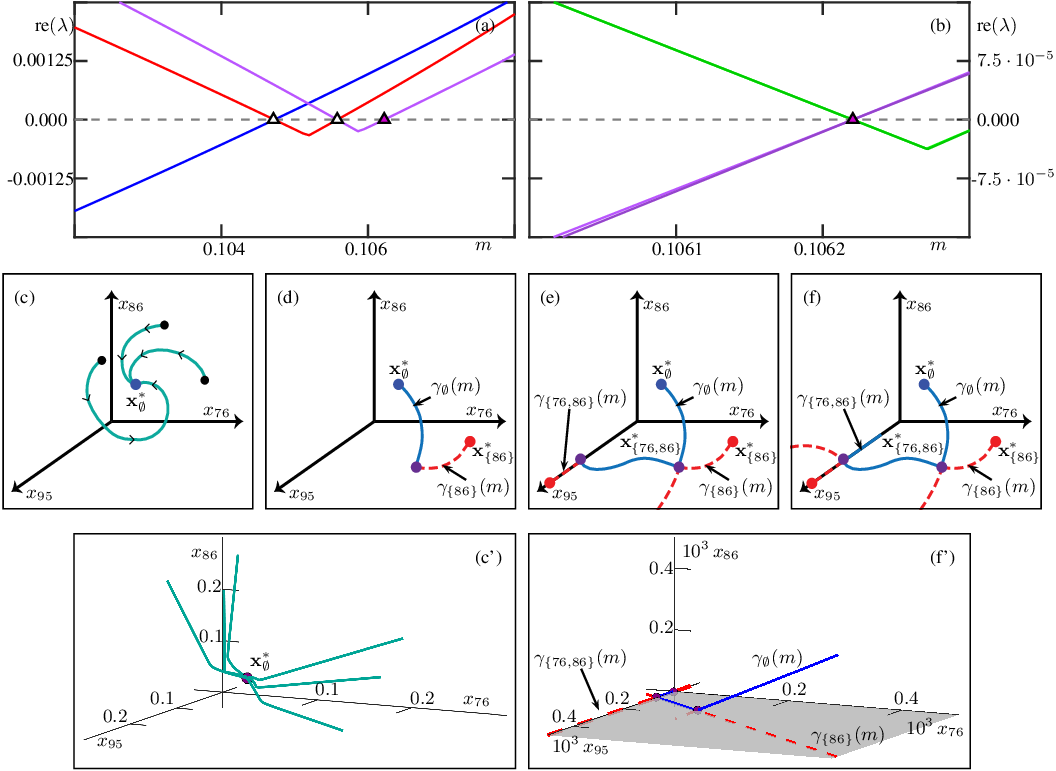}
\caption{Transcritical bifurcations and species extinction. 
(a) The largest eigenvalues of the equilibria $\mathbf{x^*_{\emptyset}}$ (blue line), $\mathbf{x}^*_{\{86\}}$ (red) and $\mathbf{x}^*_{\{86,76\}}$ (purple) as functions of $m$ for $c=0.3$ in network $A_1$.  Their signs change at $m_{\{86\}}, m_{\{86,76\}}$ and $m_{\{86,76,95,97\}}$ represented by triangles. 
(b) A magnification of (a), where shown are the two largest eigenvalues of the equilibria $\mathbf{x}^*_{\{86,76\}}$ (purple lines) and $\mathbf{x}^*_{\{86,76,95,97\}}$ (green lines). 
(c) Phase portrait in the $(x_{95},x_{76},x_{86})$-space for $m<m_{\{86\}}$, where the full coexistence equilibrium $\mathbf{x^*_\emptyset}$ (blue dot) is globally attracting. Schematic trajectories (green lines) starting from different initial conditions (black dots) are shown to converge to it; the arrows represent the direction of time. 
(d) Two equilibria $\mathbf{x^*_\emptyset}$ (blue dot) and $\mathbf{x}^*_{\{86\}}$ (red dot) move along paths $\gamma_{\emptyset}(m)$ (blue curve) and $\gamma_{\{86\}}(m)$ (red curve) with increasing $m$ until they meet in a transcritical bifurcation (purple dot) at $m=m_{\{86\}}$. A solid(dashed) line represents that the corresponding equilibrium is stable (unstable).
(e) As $m$ is further increased, $\mathbf{x}^*_{\{86\}}$ becomes stable and $\mathbf{x^*_\emptyset}$ becomes unstable and unfeasible. At the same time, another equilibrium $\mathbf{x}^*_{\{86,76\}}$ moves along a path $\gamma_{\{86,76\}}(m)$ on the $x_{95}$ axis. 
(f) Another transcritical bifurcation occurs at $m=m_{\{86,76\}}$ between  $\mathbf{x}^*_{\{86\}}$ and $\mathbf{x}^*_{\{86,76\}}$. For $m>m_{\{86,76\}}$, the equilibrium $\mathbf{x}^*_{\{86,76\}}$ is stable.
(c') The numerical counterpart of (c) showing the solution trajectories from different initial conditions converging to $\mathbf{x^*_\emptyset}$, which are obtained by numerically integrating system~\eqref{eq:LVNet} at $m=0$ and $c=0.3$  
(f') The numerical counterpart of (f) showing the paths of the three equilibria undergoing two transcritical bifurcations, obtained from the numerical continuation data at $c=0.3$.
} 
\label{fig:SketchTrans} 
\end{figure*}

Here we focus on the bifurcations generating the persistence diagram in Fig.~\ref{fig:BifNetwork1}(a3), which represents the survival and extinction of individual species when the mutualism strength $m$ is increased slowly enough for the system to reach equilibrium at each given $m$. We obtain the corresponding equilibria by tracing the attracting feasible equilibria reached from the full-coexistence equilibrium while increasing $m$ from zero within the numerical continuation framework. 

In \fref{fig:BifNetwork1}(a3),  at fixed $c=0.3$, the species labeled $86$ goes first extinct at $m = m_{\{86\}}\approx 0.10471$ while all the $107$ species coexist for $m<m_{\{86\}}$. As $m$ is increased further, the species $76$ next goes extinct at $m_{\{86,76\}}\approx  0.10558$, and then two species $95$ and $97$ go simultaneously extinct at $m_{\{86,76,95,97\}} \approx 0.10622$. 

These extinction events are characterized by the exchanges of stability among four equilibrium points, $\mathbf{x_\emptyset}^*, \mathbf{x}^*_{\{86\}},\mathbf{x}^*_{\{86,76\}}$, and $\mathbf{x}^*_{\{86,76,95,97\}}$. \Fref{fig:SketchTrans}(a) shows the largest eigenvalues, having the largest real parts, of the Jacobian at the equilibria $\mathbf{x_{\emptyset}^*}, \mathbf{x}_{\{86\}}$, and $\mathbf{x}_{\{86,76\}}$. As $m$ increases passing $m_{\{86\}}$, the largest eigenvalue of $\mathbf{x_{\emptyset}^*}$ becomes positive, while the largest eigenvalues of  $\mathbf{x}^*_{\{86\}}$ becomes negative, implying that the former becomes unstable and the latter stable. Such exchange of stability occurs also between $\mathbf{x}_{\{86\}}$ and $\mathbf{x}_{\{86,76\}}^*$ at $m_{\{86,76\}}$, and between $\mathbf{x}_{\{86,76\}}^*$ and $\mathbf{x}_{\{86,76,95,97\}}^*$ at $m_{\{86,76,95,97\}}$.  It should be noted that the two largest eigenvalues of the Jacobian at $\mathbf{x}_{\{86,76\}}^*$ become positive as $m$ increases beyond $m_{\{86,76,95,97\}}$ [Fig.~\ref{fig:SketchTrans}(b)], since they are equal due to the symmetry of the interaction matrix $\mathbb{B}^+$ regarding two species $95$ and $97$ that are connected to the same partner (plant) species given the extinction of species $86$ and $76$. 

All these transcritical bifurcations occur when the involved equilibria collide in the phase space. In the phase space projected onto the abundance of three species $(x_{95},x_{76}, x_{86})$ for ease of representation, the full-coexistence equilibrium $\mathbf{x}_\emptyset^*$, that is globally attracting as represented by the inward flow in \fref{fig:SketchTrans}(c), moves towards the $(x_{95},x_{76})$-plane, at which $x_{86}=0$, as $m$ increases until it collides at $m=m_{\{86\}}$ with another equilibrium $\mathbf{x}_{\{86\}}^*$ which is unstable as long as $m<m_{\{86\}}$ [\fref{fig:SketchTrans}(d)]. Their stability is exchanged at $m_{\{86\}}$ as described above, indicating a transcritical bifurcation. For $m>m_{\{86\}}$, the equilibrium $\mathbf{x}_{\{86\}}^*$ is stable moving in the $(x_{95},x_{76})$-plane while the former equilibrium $\mathbf{x}_{\emptyset}^*$ is unstable and unfeasible, moving in the space of $x_{86}<0$ [Fig.~\ref{fig:SketchTrans}(e)].
As $m$ is increased further,  $\mathbf{x}^*_{\{86\}}$ moves towards the $x_{95}$-axis until it coincides with another equilibrium $\mathbf{x}^*_{\{86,76\}}$ at $m_{\{86,76\}}$, which is another transcritical bifurcation. In Fig.~\ref{fig:SketchTrans}(f),  as $m$ further increases, the equilibrium $\mathbf{x}^*_{\{86,76\}}$ moves towards the equilibrium $\mathbf{x}^*_{\{86,76, 95,97\}}$, located at the origin in this projected phase space, and they collide and exchange stability at $m_{\{86,76,95,97\}}$. 

The described bifurcations are based on our numerical data showcased in \fref{fig:SketchTrans}(c') and (f'). Particularly, \fref{fig:SketchTrans}(c') illustrates the solution trajectories that converge to $\mathbf{x^*_\emptyset}$ from different initial conditions when $m=0$, the numerical counterpart of \fref{fig:SketchTrans}(c).  \Fref{fig:SketchTrans}(f') shows the continuation paths of the three equilibria, the counterpart of the ones schematically sketched in \fref{fig:SketchTrans}(f). Such transcritical bifurcations continue until the attracting equilibrium escapes at infinity, as exemplified by the termination of the persistence lines for the final 11 surviving species at $m \approx m_\infty \approx 0.5109$ at $c=0.3$ in Fig.~\ref{fig:BifNetwork1}(a3).

\subsection{Multistability arising from Hopf bifurcations} 
\label{sec:Hopf}

\begin{figure*}
  \centering
\includegraphics{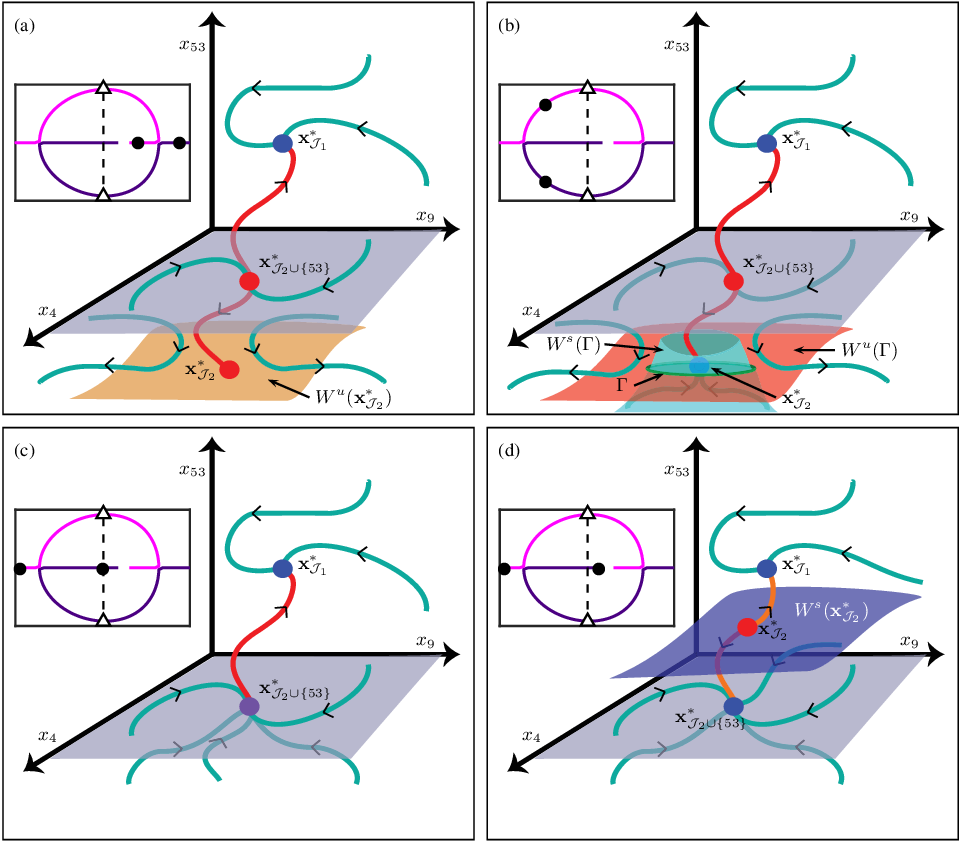}
\caption{Occurrence of Hopf bifurcation and multistability.
(a) Three equilibria in the $(x_4,x_9,x_{53})$-space for $m<m^{\rm (hopf)}$ with $c$ fixed at $0.3$. Schematic trajectories show that $\mathbf{x}^*_{\mathcal{J}_1}$ (blue dot) is an attracting equilibrium while $\mathbf{x}^*_{\mathcal{J}_2\cup\{53\}}$ (red dot) and $\mathbf{x}^*_{\mathcal{J}_2}$ (red dot) are  saddle. $\mathbf{x}^*_{\mathcal{J}_2}$ is unfeasible.  Also shown are the unstable manifold $W^u(\mathbf{x}^*_{\mathcal{J}_2})$ (orange surface) and the hyperplane $x_{53}=0$ (grey surface). Inset: Two largest eigenvalues $\lambda$'s of the Jacobian at  $\mathbf{x}^*_{\mathcal{J}_2}$ (black dots) in the complex plane with the range $|\rm{re}(\lambda)|<0.8\cdot 10^{-4}$ and $|\rm{imag}(\lambda)|<1.2\cdot 10^{-4}$. The purple and magenta curves indicate how these eigenvalues will move as $m$ increases. 
(b) At $m^{\rm (hopf)}$, a subcritical Hopf bifurcation occurs for $\mathbf{x}^*_{\mathcal{J}_2}$. For $m\gtrsim m^{\rm (hopf)}$, an unstable saddle periodic solution $\Gamma$ (green curve) is created, which has a $(s-1)$-dimensional stable manifold $W^s(\Gamma)$ (cyan surface) and two-dimensional unstable manifold $W^u(\Gamma)$ (red surface). Consequently the equilibrium $\mathbf{x}^*_{\mathcal{J}_2}$ becomes stable (blue dot). The equilibrium $\mathbf{x}^*_{\mathcal{J}_2\cup\{53\}}$  remains saddle. Inset: The two largest eigenvalues of $\mathbf{x}^*_{\mathcal{J}_2}$ are now complex conjugates with a negative real part.
(c) Two equilibria $\mathbf{x}^*_{\mathcal{J}_2\cup\{53\}}$ and $\mathbf{x}^*_{\mathcal{J}_2}$ coincide (purple dot) when $m=m^{(2)}$, which is a transcritical bifurcation. Inset: $\mathbf{x}^*_{\mathcal{J}_2}$ has one zero eigenvalue. 
(d) For $m\gtrsim m^{(2)}$, the equilibrium $\mathbf{x}^*_{\mathcal{J}_2}$ becomes a feasible saddle equilibrium with a $(s-1)$-dimensional stable manifold $W^s(\mathbf{x}^*_{\mathcal{J}_2})$ (blue surface) and one-dimensional unstable manifold $W^u(\mathbf{x}^*_{\mathcal{J}_2})$ (orange curve). On the other hand,  $\mathbf{x}^*_{\mathcal{J}_2\cup\{53\}}$ becomes stable (blue dot). Inset: $\mathbf{x}^*_{\mathcal{J}_2}$ has one positive eigenvalue. 
} 
\label{fig:TopSketchHopf} 
\end{figure*}

For large $m$, we find that multiple attracting equilibria may exist. One of these corresponds to a continuation of increasing $m$ from the full-coexistence equilibrium, while the remaining ones do not.  Here we investigate within the numerical continuation framework the underlying geometrical mechanisms for the different equilibria as shown in Fig.~\ref{fig:BifNetwork1}(a3) and (b3).

We find a new feasible attracting equilibrium $\mathbf{x}_{\mathcal{J}_2\cup \{53\}}^*$ shown in \fref{fig:BifNetwork1}(b3) emerges at $m=m^{\rm (2)}\approx 0.187267$, in addition to the one $\mathbf{x}_{\mathcal{J}_1}^*$ of Fig.~\ref{fig:BifNetwork1}(a3), when $c$ is fixed at $0.3$. The indices of extinct species in $\mathcal{J}_1$ and $\mathcal{J}_2$ are quite different [Table~\ref{tab:J1andJ2} in Appendix~\ref{app:tabSets}]. We find by numerical continuation that $\mathbf{x}_{\mathcal{J}_2\cup \{53\}}^*$ acquires stability from an unfeasible attracting equilibrium $\mathbf{x}_{\mathcal{J}_2}^*$, which has a negative abundance for species $53$, through a transcritical bifurcation at $m=m^{\rm (2)}$. The unfeasible equilibrium $\mathbf{x}_{\mathcal{J}_2}^*$ is not stable for $m<m^{\rm (hopf)}\approx 0.187251$ but becomes stable by a Hopf bifurcation, which cannot occur for any feasible equilibrium in system~\eref{eq:LVNet} with a symmetric interaction matrix $\mathbb{B}$\cite{MacArthur1970},  at $m=m^{\rm (hopf)}$ smaller than $m^{(2)}$.

We sketch the phase portrait showing these two bifurcations in the projected space of $x_{4}, x_{9}$ and $x_{53}$ in Fig.~\ref{fig:TopSketchHopf}. 
When $m<m^{\rm (hopf)}$, there exists an attracting feasible equilibrium $\mathbf{x}^*_{\mathcal{J}_1}$. In addition, there are two unstable saddle equilibria $\mathbf{x}^*_{\mathcal{J}_2\cup \{53\}}$ and $\mathbf{x}^*_{\mathcal{J}_2}$ at which the Jacobian has one and two positive eigenvalues, respectively [Fig.~\ref{fig:TopSketchHopf}(a) and the inset therein]. Here, $\mathbf{x}^*_{\mathcal{J}_2\cup \{53\}}$  is a feasible equilibrium, and  $\mathbf{x}^*_{\mathcal{J}_2}$ is an unfeasible one with a negative component for species $53$. The equilibrium $\mathbf{x}^*_{\mathcal{J}_2}$ possesses a two-dimensional unstable manifold $W^u(\mathbf{x}^*_{\mathcal{J}_2})$ (orange surface) and a $(s-2)$-dimensional stable manifold $W^s(\mathbf{x}^*_{\mathcal{J}_2})$ (not shown) which correspond to the set of initial conditions that converge backward and forward in time to $\mathbf{x}^*_{\mathcal{J}_2}$, respectively. 

As $m$ increases past $m^{\rm (hopf)}$, a saddle periodic orbit  $\Gamma$ (green curve) surrounding   $\mathbf{x}^*_{\mathcal{J}_2}$ (blue dot) is created, which possesses a two-dimensional unstable manifold $W^u(\Gamma)$ (red surface) and $(s-1)$-dimensional stable manifold $W^s(\Gamma)$ (cyan surface), indicating a subcritical Hopf bifurcation.  Therefore the unfeasible equilibrium  $\mathbf{x}^*_{\mathcal{J}_2}$ (blue dot) is now attracting [Fig.~\ref{fig:TopSketchHopf}(b)]. 

Then  $\mathbf{x}^*_{\mathcal{J}_2}$ moves towards the $(x_4, x_9)$-plane as $m$ increases further, and collides with the feasible but currently unstable equilibrium $\mathbf{x}^*_{\mathcal{J}_2\cup \{53\}}$ at $m=m^{\rm (2)}$ in a transcritical bifurcation as shown in Fig.~\ref{fig:TopSketchHopf}(c). Through this bifurcation, the equilibrium $\mathbf{x}^*_{\mathcal{J}_2\cup \{53\}}$ becomes attracting, and $\mathbf{x}^*_{\mathcal{J}_2}$ becomes a feasible saddle equilibrium with one positive real eigenvalue, as shown in Fig.~\ref{fig:TopSketchHopf}(d).  Its $(s-1)$-stable manifold $W^s(\mathbf{x}^*_{\mathcal{J}_2})$ becomes the boundary that separates the initial conditions that converge to $\mathbf{x}^*_{\mathcal{J}_1}$ and $\mathbf{x}^*_{\mathcal{J}_2\cup \{53\}}$. 

For the new feasible attracting equilibrium $\mathbf{x}^*_{\mathcal{J}_2\cup \{53\}}$, transcritical bifurcations can occur as $m$ increases representing the successive extinction of species  as described in \sref{sec:Transcritical} and shown in \fref{fig:BifNetwork1}~(b3). It should be also noted that more new equilibria can appear via the Hopf bifurcations of different unfeasible equilibria; indeed, for the network $A_1$ we observe multiple attracting equilibria arising through Hopf bifurcations as shown in \fref{fig:BifNetwork3-4} in  Appendix~\ref{app:ExtraHopf} (again for $c=0.3$). 

\subsection{Multistability arising from Degenerate Transcritical bifurcation}  \label{sec:degTrans}

\begin{figure*}
  \centering
\includegraphics{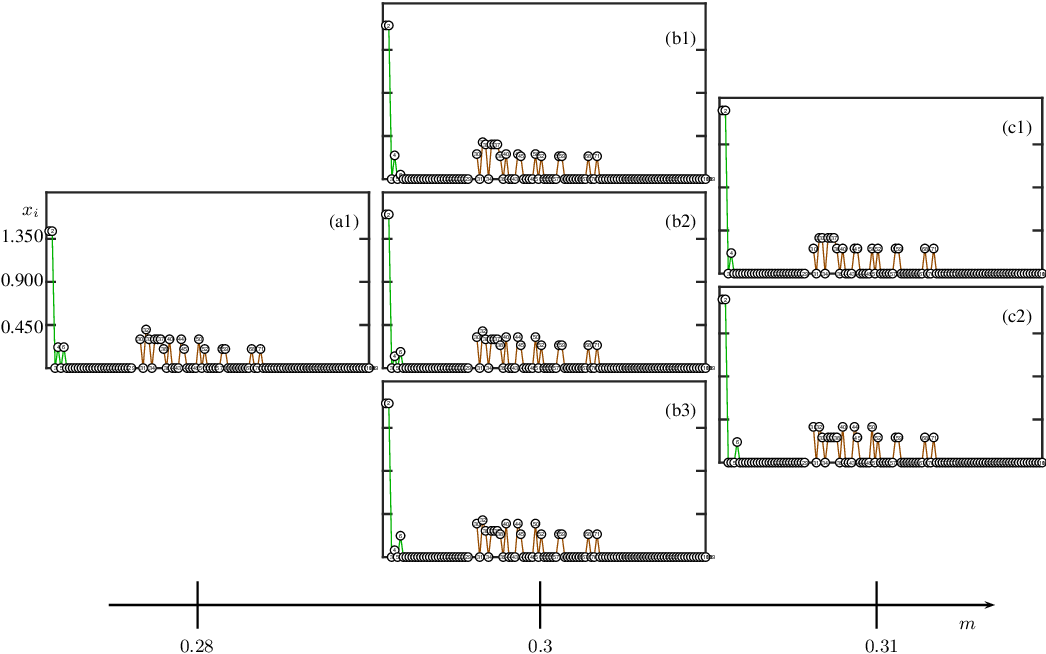}
\caption{Occurrence and consequences of a degenerate transcritical bifurcation.
(a1) Final abundances of individual species 
for any random initial condition, corresponding to the only one attracting equilibrium,  at $m=0.28$ in networks $A_2$.
(b1, b2, b3) Three examples among infinitely many different equilibria observed at $m=0.3$.
(c1, c2) Two attracting equilibria at $m=0.31$. $c$ is fixed at $0.4$ in all panels.} 
\label{fig:timeTrajMA} 
\end{figure*}

\begin{figure*}
\centering
\includegraphics{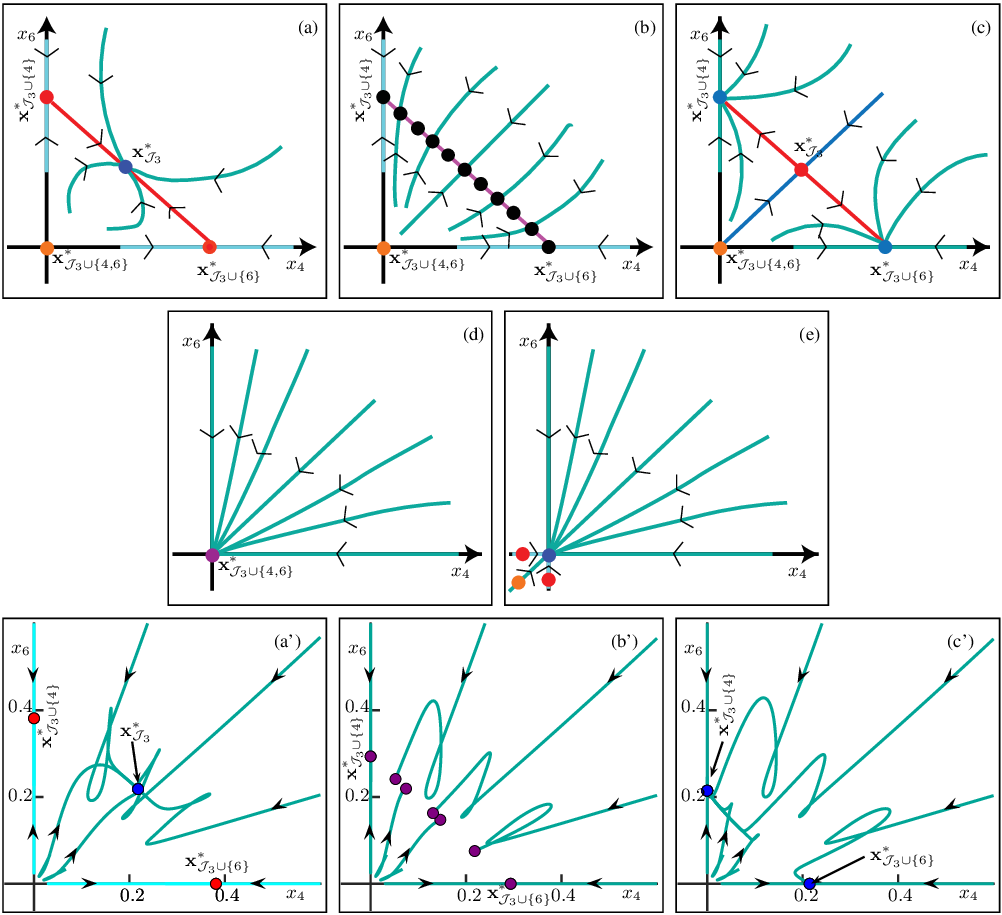}
\caption{Degenerate transcritical bifurcation and multistability. 
(a) Phase portrait in the $(x_{4}, x_{6})$-space for $m<m^{\rm (d)}$, where shown are an attracting equilibrium $\mathbf{x}^*_{\mathcal{J}_3}$ (blue dot), two saddle equilibria $\mathbf{x}^*_{\mathcal{J}_3\cup \{4\}}$ and $\mathbf{x}^*_{\mathcal{J}_3\cup \{6\}}$  (red dots), each one having only one positive eigenvalue, and $\mathbf{x}^*_{\mathcal{J}_3\cup \{4,6\}}$ (black dot) with two positive eigenvalues of the Jacobian. (b) At $m=m^{\rm (d)}$,  a degenerate transcritical bifurcation occurs to form a one-dimensional line of nonhyperbolic equilibria (line with black dots).
(c) For $m\gtrsim m^{\rm (d)}$, $\mathbf{x}^*_{\mathcal{J}\cup \{4\}}$ and $\mathbf{x}^*_{\mathcal{J}\cup \{6\}}$ are  stable and $\mathbf{x}^*_{\mathcal{J}}$ becomes an unstable saddle equilibrium. 
(d) At $m=m_{\mathcal{J}_3\cup\{4,6\}}$, four equilibria coincide.
(e) For $m\gtrsim m_{\mathcal{J}_3\cup\{4,6\}}$, the equilibrium $\mathbf{x}^*_{\mathcal{J}\cup \{4, 6\}}$ becomes stable and the other three equilibria become unfeasible. 
(a', b', c') The numerical counterparts of (a), (b) and (c), respectively, showing the solution trajectories of system ~\eref{eq:LVNet} with different initial conditions at $c=0.4$ and $m=0.28$ for (a'), $m=0.3$ for (b'), and $m=0.31$ for (c'). } 
\label{fig:SketchMA} 
\end{figure*}

Our study demonstrates that multistability can be generated also by degenerate transcritical bifurcations for the LV dynamics in system~\eqref{eq:LVNet}. Degenerate transcritical transitions occur when the reduced interaction matrix $\mathbb{B^+}$ for surviving species has zero determinant, as described in \sref{sec:Method}. By considering the network $A_2$ as a representative example,  we showcase how this mechanism arises in a real network. 

As $m$ increases from zero at fixed $c=0.4$ in network $A_2$, the full-coexistence equilibrium $\mathbf{x_\emptyset^*}$ exhibits a  series of transcritical bifurcations 
until $m$ reaches $m^{\rm (d)}=0.3$. At $m=0.29$,  the final abundances of species are obtained as given in  \fref{fig:timeTrajMA}~(a1) 
without regard to the initial conditions, implying the existence of only one attracting equilibrium. At $m=m^{\rm (d)}=0.3$, in contrast, one finds the final abundance of species $4$ and $6$ taking infinitely many different values depending on the initial conditions, 
three of which are shown in  \fref{fig:timeTrajMA}(b1) to (b3) as an example.
When $m$ is slightly larger than $m^{\rm (d)}$, which we will denote in this paper as $m\gtrsim m^{\rm (d)}$, we find only two final abundances shown in Fig.~\ref{fig:timeTrajMA}(c), one with species $4$ surviving and the other with species $6$ surviving; they cannot coexist simultaneously. The survival and extinction of other species are not different between the two solutions. We remark that during this transition, species $4$ and $6$ 
do not have the same surviving mutualistic partners;
such equivariance under the permutation of their labels~\cite{Golubitsky1985} is not necessary for this transition to occur.

The geometrical mechanism for this type of multistability, emerging through a degenerate transcritical bifurcation, is illustrated in \fref{fig:SketchMA}. At $m<m^{\rm (d)}$, there exists an attracting feasible equilibrium $\mathbf{x}^*_{\mathcal{J}_3}$ (blue dot), with $\mathcal{J}_3$ given in Table~\ref{tab:J3} in Appendix~\ref{app:tabSets}, and two saddle feasible equilibria $\mathbf{x}^*_{\mathcal{J}_3\cup \{4\}}$ and $\mathbf{x}^*_{\mathcal{J}_3\cup \{6\}}$ (red dots) [Fig.~\ref{fig:SketchMA}(a)]. At $m=m^{\rm (d)}$ [Fig.~\ref{fig:SketchMA}(b)] a degenerate transcritical bifurcation occurs  such that
$\mathbf{x}^*_{\mathcal{J}_3}$, $\mathbf{x}^*_{\mathcal{J}_3\cup \{4\}}$ and $\mathbf{x}^*_{\mathcal{J}_3\cup \{6\}}$ all become nonhyperbolic (having an imaginary eigenvalue of the Jacobian). They are connected through a line of nonhyperbolic equilibria (purple line) [Fig.~\ref{fig:SketchMA}(b)] corresponding to the one-dimensional nullspace of the reduced interaction matrix $\mathbb{B^+}$ for $\mathbf{x}^*_{\mathcal{J}}$. Different initial conditions converge to different equilibria on the line. When $m$ is larger than $m^{\rm (d)}$, the line of equilibria disappears, and the stability of equilibria is interchanged; $\mathbf{x}^*_{\mathcal{J}_3}$ becomes saddle, and $\mathbf{x}^*_{\mathcal{J}_3\cup \{4\}}$ and $\mathbf{x}^*_{\mathcal{J}_3\cup \{6\}}$ become attracting [Fig.~\ref{fig:SketchMA}(c)]. Notably, $\mathbf{x}^*_{\mathcal{J}_3}$ has a stable manifold (blue line) that separates the basin of attraction of the two new attracting equilibria; thus, certain initial conditions can converge to $\mathbf{x}^*_{\mathcal{J}_3\cup \{4\}}$ while others to $\mathbf{x}^*_{\mathcal{J}_3\cup \{6\}}$. Therefore, we see  two new attracting equilibria where a species is extinct in one and surviving in the other. \Fref{fig:SketchMA}(a'), (b') and (c') showcase with actual numerical simulations how the solutions converge to different equilibria during this transition as indicated in their schematic sketches counterparts in \fref{fig:SketchMA}(a), (b), and (c).

As $m$ is further increased after the degenerate transcritical bifurcation, the two attracting equilibria can exhibit their respective extinction sequences, keeping multistability or multistability may disappear soon, the latter of  which is the case in our example. At $m=m_{\mathcal{J}_3\cup\{4,6\}}$, as shown in Fig.~\ref{fig:SketchMA}(d), all the equilibria $\mathbf{x}^*_{\mathcal{J}_3\cup \{4\}}$, $\mathbf{x}^*_{\mathcal{J}_3\cup \{6\}}$ and $\mathbf{x}^*_{\mathcal{J}_3}$ coincide with the equilibrium $\mathbf{x}^*_{\mathcal{J}_3\cup \{4,6\}}$. When $m$ is further increased, $\mathbf{x}^*_{\mathcal{J}_3\cup \{4,6\}}$ is the only attracting equilibrium [Fig.~\ref{fig:SketchMA}(e)]. 
Given that the persistence diagrams of these two equilibria are equivalent up to $m^{\rm (d)}$ where this bifurcation occurs (blue line) [\fref{fig:BifMulti} in Appendix~\ref{app:Persistence_degenerate}],  it is only for $m^{\rm  (d)}<m<m_{\mathcal{J}_3\cup \{4,6\}}$ that both diagrams differ, meaning multistability. For $m>m_{\mathcal{J}_3\cup\{4,6\}}$, both species 4 and 6 are extinct and multistability disappears. 

It is worth mentioning that multistability from Hopf bifurcations can still occur in network $A_2$ as $m$ increases. Also, we should remark that the equilibria connected to the full-coexistence equilibrium in network $A_1$ undergo a degenerate transcritical bifurcation though it is not shown in \fref{fig:BifNetwork1}(a3).

\section{Phase diagram}
\label{sec:BifurcationDiagram}

\begin{figure*}
  \centering
\includegraphics{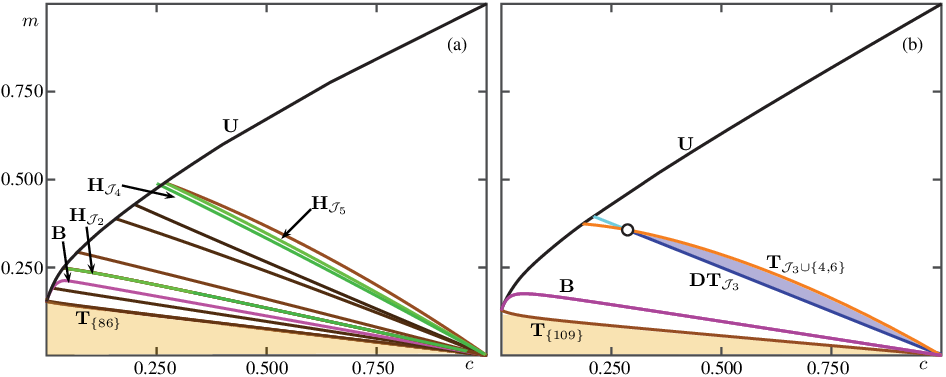}
\caption{Bifurcation diagram
for network  $A_1$ (a), and  $A_2$ (b). 
Shown are representative lines of transcritical bifurcations $\mathbf{T}$ (brown and orange), 
Hopf bifurcation $\mathbf{H}$ (green), degenerate transcritical bifurcation $\mathbf{DT}$ (blue and cyan) and the boundary of bounded asymptotic behavior $\mathbf{U}$ (black). The line $\mathbf{B}$ (purple) represents the smallest values of $m$ and $c$, where the full interaction matrix $\mathbb{B}$ in Eq.~\eqref{eq:Bmatrix} stops being negative definitive. The subindex indicates the extinct species of the equilibrium exhibiting the bifurcation. See \fref{fig:BifNetwork3-4} in Appendix~\ref{app:ExtraHopf}  for persistence diagram of $\mathbf{x}^*_{\mathcal{J}_4}$ and $\mathbf{x}^*_{\mathcal{J}_5}$.} \label{fig:BifPlaneA1} 
\end{figure*}

The numerical continuation approach allows one to obtain the loci of bifurcations in the parameter space and thereby identify the parameter regimes or {\it phases} characterized by distinct sets of extinct species and multistability. Such a phase diagram is presented in \fref{fig:BifPlaneA1}(a) for network $A_1$, in which the first transcritical bifurcation line labeled as $\mathbf{T}_{\{ 86\}}$ (light brown) is drawn, distinguishing the full-coexistence phase  $\mathbf{x^*_\emptyset}$ and the phase with species 86 extinct $\mathbf{x}_{\{86\}}^*$. Above the  line $ \mathbf{T}_{\{ 86\}}$ in the $(c,m)$-plane, a series of transcritical bifurcation lines are located, representing further extinctions, some of which we present by darker brown lines. We also illustrate representative Hopf bifurcations, represented by green lines, which correspond to the mechanism of creating multistability; indeed, really close to these lines we find transcritical bifurcations (not shown) that mark the onset of multistability in system~\eref{eq:LVNet}, above which multiple attracting feasible equilibria coexist. 

As stated earlier, the full coexistence equilibrium $\mathbf{x}^*_\emptyset$ is globally attracting if all of its components are positive and the full interaction matrix $\mathbb{B}$ is negative definite. To understand when this global property is lost, we obtain the boundary where $\mathbb{B}$ stops being negative definitive, which is the line labeled as $\mathbf{B}$ in \fref{fig:BifPlaneA1}(a). Below this line,  $\mathbb{B}$ is negative definite. As it turns out that the line $\mathbf{B}$ is located above the first transcritical bifurcation line $\mathbf{T}_{\{ 86\}}$, the full-coexistence equilibrium $\mathbf{x}_{\emptyset^*}$ is a global attractor and there is no multistability in the orange colored region in  \fref{fig:BifPlaneA1}(a).

We find that all these bifurcation lines (phase boundaries) emanate from the point $(c,m)=(1,0)$. System~\eqref{eq:LVNet} is highly degenerate at $(c,m)=(1,0)$ as it is reduced to ${dx_i \over dt} = x_i( 1 - \sum_{p'=1}^{n_p} x_{p'})$  for all plants $i$, and ${dx_j \over dt} = x_j( 1 - \sum_{a'=1}^{n_a} x_{a'})$ for all animals $j$, leading to a family of nonhyperbolic equilibria in the $(s-2)$-dimensional hyperplane given by the intersection of the hyperplanes $\sum_p x_p=1$ and $\sum_a x_a =1$. 

We also compute the boundary for unbounded abundances $\mathbf{U}$,  where the equilibria disappear at infinity. In \fref{fig:BifPlaneA1}(a), it is shown that the transcritical bifurcation lines,  which correspond to the successive extinctions of the full coexistence equilibrium, terminate at the black line $\mathbf{U}$.  In the persistence diagrams in Figs.~\ref{fig:BifNetwork1}(a3), ~\ref{fig:BifNetwork1}(b3), ~\ref{fig:BifNetwork3-4}(a), and ~\ref{fig:BifNetwork3-4}(c), the last surviving species find their abundances diverging at certain values of $m$. To obtain this boundary numerically, we first need to identify the set of the last surviving species that cannot undergo further transcritical bifurcations. To do so, we fix $m$ at a value before the equilibrium diverges, and then we apply continuation on $c$ until it is equal to one. Then, we fix $c$ at one and then continue on $m$. During this procedure, we jump to new attracting equilibria if further transcritical bifurcation occurs. In this way, we can find the last set of surviving species when $c$ is one. After identifying this set, we can restart the continuation procedure by varying $m$ with $c=0$. Since only mutualism occurs, $\mathbf{x^*_\emptyset}$ cannot exhibit extinction, and it will grow unbounded as $m$ increases. We then choose a value of $m^{\rm (u)}$ for which one of the abundances of the set of the last surviving species becomes large, e.g., $x_{\rm surv}=10^4$. We then obtain the equilibrium by solving the zero-problem in Eq.~\eref{eq:LVNetZeros} under the constraint 
$x_{\rm surv} = 10^4$,
and trace by numerical continuation this equilibrium and its transcritical bifurcations while simultaneously varying $c$ and $m$. This leads us to approximate the terminus of the transcritical bifurcation lines numerically.  We remark that for each attracting equilibrium that emerges through the Hopf bifurcation mechanism, there exists an associated boundary of unbounded abundance (not shown). This demonstrates that compactification techniques can be used to study these bifurcations \cite{Messias2009, Messias2011, Giraldo2017, Giraldo2020I, Matsue2017}.

The phase diagram for network $A_1$ in \fref{fig:BifPlaneA1}~(a) implies much. As the transcritical bifurcation lines, emanating from the degenerate point $(1,0)$,  have negative slopes in the $(c,m)$-plane, one can expect that an increase in mutualism or competition can lead to species extinction; If the competition (mutualism) strength is increased, then the mutualism (competition) strength should be decreased sufficiently to avoid species extinction. It also suggests that as competition in the network intensifies, multistability can arise even with weak mutualism. This indicates that biodiversity scenarios in highly competitive uniform networks are quite fragile, even when weak nonuniform mutualistic dynamics are introduced.

The phase diagram for network $A_2$ is also shown in Fig.~\ref{fig:BifPlaneA1}(b), in which we present i) the first transcritical bifurcation line (brown) for the extinction of species 109, ii) the degenerate transcritical bifurcation line (blue), and iii) the transcritical bifurcation (orange) that distinguish the regime of either species 4 or 6 surviving and that of both extinct. Like for network $A_1$, all these lines emanate from the degenerate point $(c,m)=(1,0)$, implying its role as the organizing center for the transcritical, Hopf, and degenerate transcritical bifurcation lines that create extinction and multistability. 

Interestingly, multistability associated with the survival and extinction of species $4$ and $6$ occurs only when the competition strength is large enough for the degenerate transcritical bifurcation to occur at a smaller value of $m$ than that for the transcritical bifurcation. For instance, at a small value of competition strength, $c=0.250$,  no multistability arises; the phase portrait transits from Fig.~\ref{fig:SketchMA}(a) to  (e) via (d), without (b) or (c), at $m\approx 0.3645$.   That is, species $4$ and $6$ become extinct simultaneously, and there is no range of $m$ where one of them exclusively exists. We also present in \fref{fig:BifPlaneA1}(b) the locus where the equilibria disappear at infinity (black line). 

\section{Summary and Discussion}\label{sec:Conclusion}

In this work, we considered a model community of two species groups— plants and pollinators—under uniform intragroup competition and empirical heterogeneous intergroup mutualism. Through a geometrical and numerical continuation approach, we provided an exact numerical method to systematically pinpoint the loci of the extinction of individual species in the parameter space for two empirical networks. Our approach has pushed us to understand emergent phenomena in structured communities, whereas most theoretical and numerical approaches have so far studied the random, unstructured case. 
 
Most importantly, through our approach, we were able to identify the specific bifurcations responsible for the creation of multistability in the empirical networks we studied.  These include subcritical Hopf bifurcations and degenerate transcritical bifurcations, which serve as mechanisms for creating multiple attracting equilibria. The multistability that they generate exhibit different characteristics.  Specifically, the Hopf bifurcations lead to the formation of new attractive equilibria where the composition of surviving species may not overlap with that of pre-existing equilibria. Conversely, degenerate transcritical bifurcations induce multistability by generating new attracting equilibria that differ only in the finite set of species that survive.

Given that the present study is limited to selected large empirical networks and simple model dynamics, it is desirable to extend our findings to a broader range of interacting communities. This extension should include a thorough investigation of the structural and dynamical conditions for the emergence of multistability, specifically through mechanisms such as Hopf and degenerate transcritical bifurcations. A deeper theoretical exploration of how empirical structural characteristics of ecological communities, such as degree heterogeneity, nestedness, non-uniform interaction strengths, and higher-order interactions, affect bifurcations and multistability is important. Furthermore, a comprehensive analysis of the stability of all equilibrium points in smaller systems, utilizing numerical continuation methods developed in this study, will enhance our understanding of these complex dynamics.

\begin{acknowledgments} 
This work was supported by  KIAS Individual Grants [Nos CG086101 (A.G.) and CG079902 (D.-S.L.)] at Korea Institute for Advanced Study.
\end{acknowledgments}

\bibliography{GL_LVNetworks}

\clearpage
\appendix
\onecolumngrid

\section{Global attractiveness   of the full coexistence equilibrium in case of no mutualism} \label{app:Goh}
\begin{theorem}[Goh's theorem \cite{Goh1979}]\label{thm:Goh}
 Suppose that the Lotka-Volterra system  $\dot{\mathbf{x}}=\mathbb{X} \lp( \boldsymbol{r} + \mathbb{C}\mathbf{x} \rp)$, for $\boldsymbol{r} \in \R^s$ and $\mathbb{C}\in M_s(\R)$, has a unique interior equilibrium $x^* = -\mathbb{C}^{-1}\boldsymbol{r} \in \R^{s}_{>0}$. Then this equilibrium is globally attracting on $\R^{s}_{>0}$ if there exists a diagonal matrix $\mathbb{D}>0$ (every entry is positive) such that $\mathbb{C}\mathbb{D} + \mathbb{D}\mathbb{C}^T$ is negative definite.
\end{theorem}

Thus, we have the following:
\begin{cor} \label{cor:GlobalAttr}
The matrix $\mathbb{B}_0:=-(1-c)\mathbb{I} - c\mathbb{J}$ is negative definite. Thus, the equilibrium $\mathbf{x^*_\emptyset}$ is globally attracting equilibrium in $\R^{s}_{>0}$ for system~\eref{eq:LVNet}  when $m=0$ and $0\leq c<1$.
\end{cor}
\begin{proof}
  For an arbitrary vector $\mathbf{x}=(x_1,x_2,...,x_s)\in \R^s$, it follows from
  \begin{equation}
    \mathbb{B}_0 = \begin{pmatrix}
      -(1-c)\mathbf{I}_{n_p} - c \mathbf{J}_{n_p} & \mathbf{0}_{n_p,n_a} \\[1mm]
      \mathbf{0}_{n_a,n_p} & -(1-c)\mathbf{I}_{n_a} - c \mathbf{J}_{n_a}
    \end{pmatrix}
  \end{equation}
that $\mathbf{x}^T\mathbb{B}_0\mathbf{x}<0$ (negative definiteness) is equivalent to $$\mathbf{x}_{p}^T\lp[(1-c)\mathbf{I}_{n_p} + c \mathbf{J}_{n_p}\rp]\mathbf{x}_{p}+\mathbf{x}_a^T\lp[(1-c)\mathbf{I}_{n_a} + c \mathbf{J}_{n_a}\rp]\mathbf{x}_{a}>0,$$
where $\mathbf{x}_{p}=(x_1,x_2,...,x_{n_p})$ and $\mathbf{x}_{a}=(x_{n_p+1},x_{n_p+2},...,x_{n_a})$. Notice that
\begin{align}
  \mathbf{x}_{p}^T\lp[(1-c)\mathbf{I}_{n_p} + c \mathbf{J}_{n_p}\rp]\mathbf{x}_{p} &= (1-c) \mathbf{x}_{p}^T \mathbf{x}_{p} +c \mathbf{x}_{p}^T \mathbf{J}_{n_p} \mathbf{x}_{p} \nonumber \\ 
  &= (1-c) \sum^{n_p}_{i=1} x^2_i + c \sum^{n_p}_{i=1}\sum^{n_p}_{j=1} x_ix_j = (1-c) \sum^{n_p}_{i=1} x^2_i + c \lp( \sum^{n_p}_{i=1} x_i\rp)^2.\label{eq:IneqGoh1}
\end{align}
It follows that Eq.~\eref{eq:IneqGoh1} is strictly positive since $0 \leq c<1$. By a symmetrical argument, it is clear that $\mathbf{x}_a^T\lp[(1-c)\mathbf{I}_{n_a} + c \mathbf{J}_{n_a}\rp]\mathbf{x}_{a}$ is also strictly positive, thus we have that
$$\mathbf{x}_{p}^T\lp[(1-c)\mathbf{I}_{n_p} + c \mathbf{J}_{n_p}\rp]\mathbf{x}_{p}+\mathbf{x}_a^T\lp[(1-c)\mathbf{I}_{n_a} + c \mathbf{J}_{n_a}\rp]\mathbf{x}_{a}>0$$
holds for any arbitrary vector $\mathbf{x}$; hence, $\mathbb{B}_0$ is negative definite. Now by making $\mathbb{D}=\mathbb{I}$ and $\mathbb{C}=\mathbb{B}_0$ then, by Goh's theorem, $\mathbf{x^*_\emptyset}$ is globally attracting equilibrium in $\R^{s}_{>0}$ for system~\eref{eq:LVNet}  when $m=0$ and $0 \leq c<1$.
\end{proof}

\section{Sets of extinction species $\mathcal{J}_1, \mathcal{J}_2$, and $\mathcal{J}_3$}\label{app:tabSets}
In this section, we provide in Table~\ref{tab:J1andJ2} and ~\ref{tab:J3} the list of extinct species $\mathcal{J}_1, \mathcal{J}_2$, and $\mathcal{J}_3$ for the  equilibria studied in the main text. 

\begin{table*} 
\begin{tabular}{|l|c|}
\hline
\multicolumn{1}{|c|}{} & 6, 13, 14, 17, 18, 19, 21, 24, 25, 27, 28, 29, 30, 31, 32, 34, 35,  36, 37, 38, 39, 40, 41, 42, 43, 50,  51,  52,  54,      \\
$\mathcal{J}_1$                     & 57,  58,  59,  60, 61,  62,  63,  64,  65,  66,  67,  70,  72,  73,  74,  75,  76,  77,  80,  82,  84,  85,  86,  87,  88,  89,  90,  \\
                       & 91,  92,  93,  94,  95,  96, 97,  98,  99, 100, 102, 104, 105, 106, 107            \\ \hline
                       & 13, 14, 18, 19, 22, 24, 25, 26, 27, 28, 30, 32, 34, 35, 36, 37, 38, 39, 40, 41, 42, 43, 51,  56,  57, 60,  62,  67,  68,   69,              \\
$\mathcal{J}_2$                     & 71,  72,  73,  74, 75,  76,  77,  78,  79,  80,  81,  82,  83,  84,  85,  86,  88, 89,  90,  92,  93,  94,  95,  96,  97,  99, 100, 101,         \\
                       & 102, 103,  104, 105, 106, 107          \\ \hline
\end{tabular}
\caption{Sets of extinct species $\mathcal{J}_1$ and $\mathcal{J}_2$ of the equilibria shown in \fref{fig:TopSketchHopf}.
}\label{tab:J1andJ2}
\end{table*}

\begin{table*}
\begin{tabular}{|l|c|}
\hline
\multicolumn{1}{|c|}{} & 3,  5,  7,  8,  9, 10, 11, 12, 13, 14, 15, 16, 17, 18, 19, 20, 21, 22, 23, 24, 25, 26, 27, 28, 2931,  34,  39,  41,  42,  43,  46,  47,  48,  49, \\
$\mathcal{J}_3$                     & 51,  53,  54, 55,  56,  57,  60,  61,  62,  63,  64,  65,  66,  67,  69,  70, 72,  73,  74,  75,  76,  77,  78,  79,  80,  81,  82,  83,  84,     \\
                       & 85,  86,  87,  88,  89,  90,  91,  92,  93,  94,  95,  96,  97,98,  99, 100, 101, 102, 103, 104, 105, 106, 107, 108, 109, 110                     \\ \hline
\end{tabular}
\caption{Sets of extinct species $\mathcal{J}_3$ of the equilibrium exhibiting the degenerate transcritical bifurcation for network $A_2$ at $c=0.4$ and $m = 0.3$ shown in \fref{fig:SketchMA}}\label{tab:J3}
\end{table*}

\section{Additional equilibria arising through Hopf bifurcations in network $A_1$}\label{app:ExtraHopf}
In this section, we present the persistence diagram of two additional attracting equilibria that emerge from Hopf bifurcations in network $A_1$. Both families of equilibria shown in \fref{fig:BifNetwork3-4} exist for a narrow range of $m$ and disappear at infinity as $m$ increases.  

\begin{figure}
\centering
\includegraphics{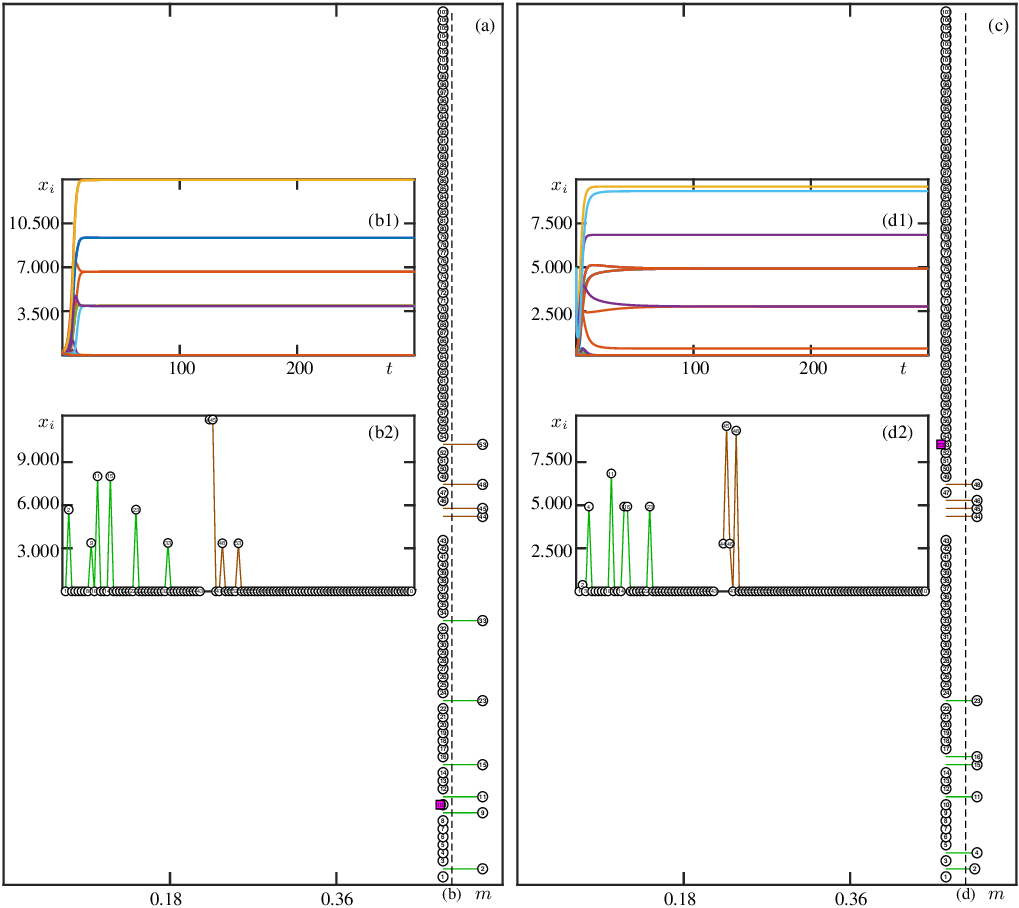}
\caption{
 Additional multistability driven by Hopf bifurcations
in network $A_1$.
Panels (a) and (c) show the persistence diagrams computed starting from unfeasible equilibria, $\mathbf{x}^*_{\mathcal{J}_4}$ and $\mathbf{x}^*_{\mathcal{J}_5}$ appearing in \fref{fig:BifPlaneA1}, respectively, that become stable through Hopf bifurcations. The purple squares indicate the species whose abundances are negative until transcritical bifurcations occur turning feasible equilibria stable.
The vertical dashed lines in panels~(a) and (c) indicate the values of $m$ selected to create panels~(b1), (b2),(d1), and (d2). Panels~(b1) and (d1) show the time trajectories of the abundances $x_i$'s for the initial conditions converging to the indicated equilibria in panels~(a) and (c), respectively, while panels~(b2) and (d2) illustrate the corresponding final abundances of species. } \label{fig:BifNetwork3-4} 
\end{figure}

\section{Persistence diagrams displaying degenerate transcritical bifurcations in network $A_2$}
\label{app:Persistence_degenerate}

In this section, we present the persistence diagrams displaying a degenerate transcritical bifurcation in network $A_2$.

\begin{figure*}
\centering
\includegraphics{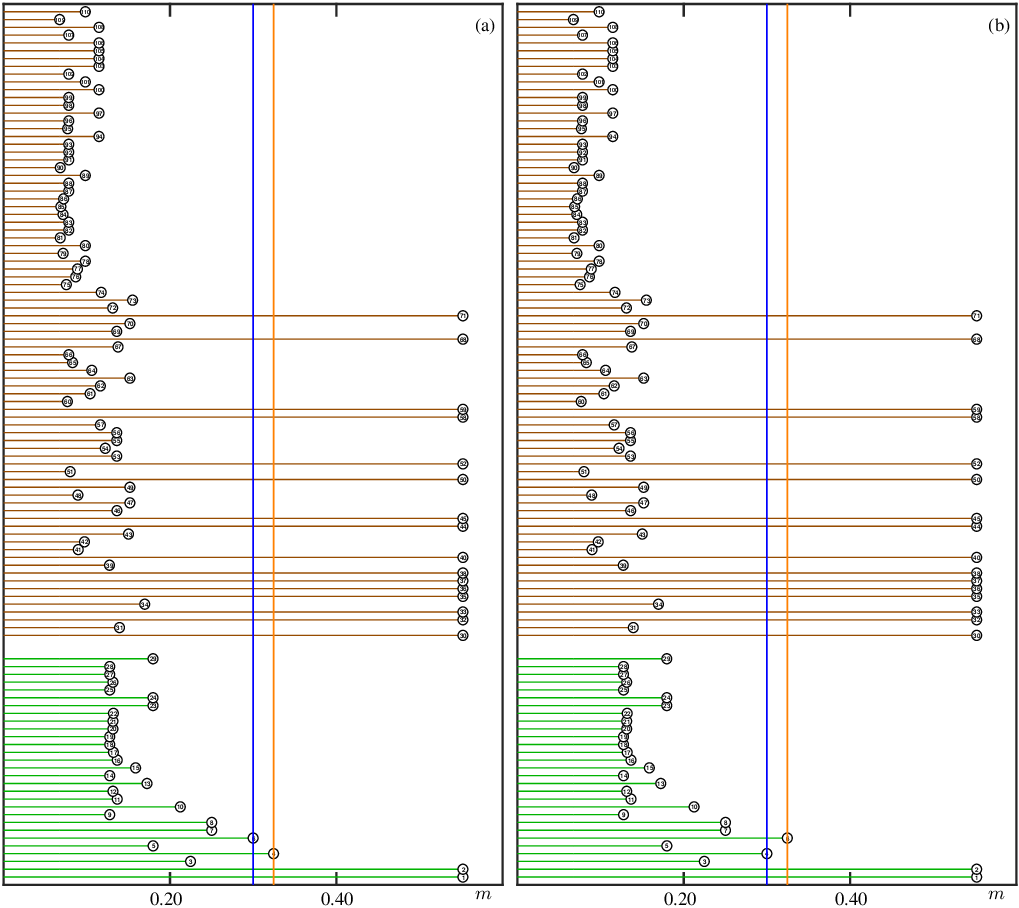}
\caption{Two different persistence diagrams, (a) and (b),  at $c=0.4$ in network $A_2$. They are computed starting from the full-coexistence equilibrium $\mathbf{x_\emptyset^*}$ at $m=0$. The vertical lines are drawn at $m=m^{\rm (d)}$ (blue) and $m=m_{\mathcal{J}_3\cup\{4,6\}}$ (orange), representing the degenerate transcritical bifurcation that creates multistability, and the transcritical bifurcation that terminates it as described in Sec.~\ref{sec:degTrans}.
} \label{fig:BifMulti} 
\end{figure*}

\end{document}